\documentclass{jnmp}

\usepackage{amsmath}
\usepackage[english]{babel}
\usepackage{bm}
\usepackage{cite}

\JNMPnumberwithin{equation}{section}

\newtheorem{theorem}{Theorem}
\newtheorem{proposition}{Proposition}

\theoremstyle{definition}
\newtheorem{definition}{Definition}
\newtheorem{example}{Example} 
\newtheorem{remark}{Remark}

\DeclareMathOperator{\Complex}{\mathbb{C}}
\DeclareMathOperator{\Real}{\mathbb{R}}
\DeclareMathOperator{\Integer}{\mathbb{Z}}
\DeclareMathOperator{\Tr}{Tr} \DeclareMathOperator{\res}{res}
\DeclareMathOperator{\spanOp}{span} \DeclareMathOperator{\ad}{ad}
\DeclareMathOperator{\Sym}{Sym}

\begin{document}

%
\renewcommand{\evenhead}{J Bernatska and P Holod}
\renewcommand{\oddhead}{On Separation of Variables for Integrable Equations of Soliton
Type}

%
\thispagestyle{empty}

\FirstPageHead{14}{3}{2007}{\pageref{firstpage}--\pageref{lastpage}}{Article}

\copyrightnote{2007}{J Bernatska and P Holod}

\Name{On Separation of Variables for Integrable Equations of Soliton
Type}

\label{firstpage}

\Author{Julia Bernatska~$^\dag$ and Petro Holod~$^\ddag$ and }

\Address{$^\dag$ The National University of 'Kiev-Mohyla Academy', Ukraine \\
~~E-mail: jnb@ukma.kiev.ua\\[10pt]
$^\ddag$ The National University of 'Kiev-Mohyla Academy', Ukraine \\
~~Bogolyubov Institute for Theoretical Physics, Ukraine \\
~~E-mail: holod@ukma.kiev.ua}

\Date{Received August 11, 2006; Accepted October 6, 2006}

\begin{abstract}
We propose a general scheme for separation of variables in the
integrable Hamiltonian systems on orbits of the loop algebra
$\mathfrak{sl}(2,\Complex)\times \mathcal{P}(\lambda,\lambda^{-1})$.
In particular, we illustrate the scheme by application to modified
Korteweg---de Vries (MKdV), sin(sinh)-Gordon, nonlinear
Schr\"{o}dinger, and Heisenberg magnetic  equations.
\end{abstract}

\section*{Introduction}

Let us make a brief review of the problem.

After the fundamental paper \cite{Novikov}, B.~Dubrovin
(see~\cite{Dubrovin75}) proposed a separation of variables for
finite gap KdV system. B.~Dubrovin shows that the poles of an
appropriately normalized Baker---Akhiezer function for the auxiliary
linear spectral problem are the separation variables. The new
variables evolve on a hyperelliptic Riemannian surface $\mathcal{R}$
of genus $g$. The genus coincides with the number of degrees of
freedom of the finite gap phase space.

In the papers~\cite{Kozel, Its, Holod78, Prykarpatski} a separation
of variables is realized for $\sin$-Gordon equation, nonlinear
Schr\"{o}dinger equation, and the classic Thirring model. The case
of $\sin$-Gordon equation appears to be completely similar to the
KdV system. However, the cases of nonlinear Schr\"{o}dinger equation
and Thirring model have a distinction: \emph{the number of degrees
of freedom is greater by one than the genus of the corresponding
spectral curve}. Here the papers \cite{Its, Holod78, Prykarpatski}
suggest the separation of variables on a reduced phase space. Later,
the complex Liouville torus of nonlinear Schr\"{o}dinger equation
was proven to be the generalized Jacobian of a singular Riemannian
surface (see \cite{Previato}).

The ideas of the early papers on the integration of finite gap
systems were generalized by E.~Sklyanin~\cite{Sklyanin92,
Sklyanin95} and partly extended to the quantum integrable
models~\cite{Sklyanin92hep, Smirnov}.

At the beginning of the 90s a new technique of separation of
variables appeared that effectively uses bi-hamiltonian, or
multi-hamiltonian, properties of integrable systems,
see~\cite{Antonowicz, Blaszak98, Blaszak99, Blaszak00}
and~\cite{Magri91, Magri00R, Magri00T, Falqui}. The main result in
this direction is the diagonalization of recursive Nijenhuis
operator. In the papers~\cite{Magri91, Magri00T} the method is
applied to KdV and Boussinesq hierarchies, and classical
finite-dimensional systems.

The papers \cite{Veselov82, Veselov84, Dubrovin82} investigate the
connection between the problem of separation of variables and the
parametrization of compact tori by symmetric products of Riemannian
surfaces. According to~\cite{Veselov82, Veselov84}, \emph{if a
change of variables reduces Liouville 1-form to a sum of meromorphic
differentials on the corresponding Riemannian surface, then we say
that the new variables are the separation variables.}

We propose a method of separation of variables for integrable
Hamiltonian systems that is connected with the orbit structure of
affine Lie algebras. The fact that finite gap phase space of an
integrable soliton hierarchy has an orbital structure was
established in~\cite{Holod85, Holod83}.\footnote{The results
of~\cite{Holod85, Holod83}  are partially covered
by~\cite{Harnad92}. However the authors of~\cite{Harnad92}  took no
notice of the remarkable \emph{duality} between pairs of soliton
equations: MKdV and sin-Gordon equations, KdV and Liouville
equations, nonlinear Schr\"{o}dinger and Heisenberg magnetic
equations, etc. The duality is evident if one uses the orbital
approach. The pairs of dual equations have common  Liouville torus
and separation variables. Sometimes, there exists a gauge
equivalence between the equations of a pair, and the equivalence
extends to the total infinite phase space \cite{Takhtajan}.} The
Hamiltonian systems in question obey the equations of Lax type, and
hence the separation variables are points on the corresponding
\emph{spectral curve}. Note that such systems are multi-Hamiltonian,
which connects our results with the results of~\cite{Magri91,
Magri00R, Magri00T, Falqui}.

In Sections 1 and 3 we reproduce the key results from
\cite{Holod83,Holod85} about finite gap phase spaces for integrable
equations as orbits of loop algebra. We illustrate our scheme by the
examples of modified Korteweg-de Vries (MKdV) system,
sin(sinh)-Gordon equation, nonlinear Schr\"{o}dinger equation, and
Heisenberg magnetic chain.

This paper is organized as follows.  Sections 1 and 2 are devoted to
MKdV system and sin(sinh)-Gordon equation. In Section 1 we construct
adjoint Poisson spaces and define the orbits regarded as phase
spaces for MKdV system and sin(sinh)-Gordon equation. The
construction is discussed in more detail in \cite{Bern}. In Section
2 we describe the scheme for separation of variables and illustrate
it by application to MKdV system and sin(sinh)-Gordon equation.
We show that the separation of variables
is achieved on both orbits simultaneously. In Sections 3 we
construct adjoint Poisson spaces and define the orbits regarded as
phase spaces for nonlinear Schr\"{o}dinger equation and Heisenberg
magnetic chain. In Sections 4 and 5 we similarly consider separation
of variables for nonlinear Schr\"{o}dinger equation and Heisenberg
magnetic chain, accordingly.

\textbf{Acknowledgements.} The authors are grateful to participants
of the scientific seminar `Integrable Hamiltonian Systems and
Solitons' T.~Skrypnyk, D.~Leikin, N.~Yorgov for useful remarks and
discussions.

\section{Phase spaces for MKdV system and sin-Gordon equation
as orbits in  $\mathfrak{sl}(2,\Complex)\otimes
\mathcal{P}(\lambda,\lambda^{-1})$}

First, let us recall some constructions from \cite{Holod85,
Holod83}. Take the algebra $\mathfrak{sl}(2,\Complex)$ with the
basis
\begin{equation*}
  H=\begin{pmatrix} \frac{1}{2}&0 \\ 0&
  -\frac{1}{2}\end{pmatrix},\qquad
  X=\begin{pmatrix} 0&1 \\ 0&0\end{pmatrix},\qquad
  Y=\begin{pmatrix} 0&0 \\ 1&0\end{pmatrix}.
\end{equation*}
Suppose $\mathcal{P}(\lambda,\lambda^{-1})$ is the algebra of
Laurent polynomials in $\lambda$. Denote by
$\widetilde{\mathfrak{g}}$ the algebra
$\mathfrak{sl}(2,\Complex)\otimes
\mathcal{P}(\lambda,\lambda^{-1})$. Then
\begin{equation}\label{basisMKdV&SG}
  H^{2m} = \lambda^m H,\qquad X^{2m+1}=\lambda^m X,\qquad
  Y^{2m+1}=\lambda^{m+1}Y
\end{equation}
is a \emph{basis} in $\widetilde{\mathfrak{g}}$.

Consider the operator
\begin{equation*}
  d = 2\lambda\frac{d}{d\lambda} + \ad_{H};
\end{equation*}
we call it the \emph{operator of principal grading}. It is easy to
prove that the basis elements (\ref{basisMKdV&SG}) are the
eigenvectors of $d$. We call the eigenvalues of $d$ the
\emph{degrees}. The superscripts in the lefthand sides of
(\ref{basisMKdV&SG}) indicate the corresponding principal degrees of
the basis elements. By $\mathfrak{g}_l$, $l\in \Integer$, denote an
eigenspace of principal degree $l$. It is evident that
\begin{equation*}
  \mathfrak{g}_{2m}=\spanOp_{\Complex} \{ H^{2m}\},\qquad
  \mathfrak{g}_{2m+1}=\spanOp_{\Complex} \{X^{2m+1}, Y^{2m+1}\}.
\end{equation*}

Decompose $\widetilde{\mathfrak{g}}$ into two subalgebras
\begin{equation*}
  \widetilde{\mathfrak{g}}_+ = \sum_{l\geqslant
  0}\mathfrak{g}_l,\qquad
  \widetilde{\mathfrak{g}}_- = \sum_{l < 0}\mathfrak{g}_l,\qquad
  \widetilde{\mathfrak{g}} = \widetilde{\mathfrak{g}}_+ +
  \widetilde{\mathfrak{g}}_-.
\end{equation*}
Further, consider the \emph{$\ad$-invariant bilinear forms}
\begin{equation}\label{bilinForms}
  \langle A(\lambda), B(\lambda) \rangle_{k} = \res \lambda^{-k-1} \Tr A(\lambda)
  B(\lambda), \quad A(\lambda),\ B(\lambda)\in
  \widetilde{\mathfrak{g}}, \quad k\in \Integer.
\end{equation}

We use the forms to define the spaces dual to
$\widetilde{\mathfrak{g}}_+$ and $\widetilde{\mathfrak{g}}_-$.
\begin{example}\label{E:k_-1}
Let $k=-1$. We have
\begin{equation}\label{AdjSpaces}
  (\widetilde{\mathfrak{g}}_-)^{\ast} = \widetilde{\mathfrak{g}}_+ +
  \mathfrak{g}_{-1},\qquad  (\widetilde{\mathfrak{g}}_+)^{\ast}=
  \sum_{l\leqslant -2} \mathfrak{g}_l,
\end{equation}
where $(\widetilde{\mathfrak{g}}_-)^{\ast}$ and
$(\widetilde{\mathfrak{g}}_+)^{\ast}$ contain only the nonzero
functionals on $\widetilde{\mathfrak{g}}_{\pm}$.
\end{example}
\begin{example}\label{E:k_N}
Let $k=N\geqslant 0$. Then
\begin{equation*}
  (\widetilde{\mathfrak{g}}_-)^{\ast} = \sum_{l\geqslant 2N+1}
  \mathfrak{g}_{l},\qquad  (\widetilde{\mathfrak{g}}_+)^{\ast}=
  \sum_{l\leqslant 2N} \mathfrak{g}_l.
\end{equation*}
\end{example}
Fix $N\geqslant 0$. Consider $M^{N+1} \subset
\widetilde{\mathfrak{g}}$,  where an element $\widehat{\mu}(\lambda)
\in M^{N+1}$ has the form
\begin{equation*}
  \widehat{\mu}(\lambda) = \begin{pmatrix}
  \alpha(\lambda) & \beta(\lambda) \\ \gamma(\lambda) &
  -\alpha(\lambda) \end{pmatrix}
\end{equation*}
with
\begin{equation*}
  \alpha(\lambda)=\sum_{m=0}^N \lambda^m
  \alpha_{2m},\qquad \beta(\lambda) = \sum_{m=0}^{N+1} \lambda^{m-1}
  \beta_{2m-1},\qquad \gamma(\lambda) = \sum_{m=0}^{N+1} \lambda^{m}
  \gamma_{2m-1}.
\end{equation*}
We call $M^{N+1}$ the $N$-\emph{gap sector of
$\widetilde{\mathfrak{g}}$}, or shortly the \emph{finite gap
sector}.

Because the factor-algebra
$\widetilde{\mathfrak{g}}_-/\sum_{l\leqslant -2N-4} \mathfrak{g}_l$
acts effectively on $M^{N+1}$, the coadjoint action of
$\widetilde{\mathfrak{g}}_-$ with respect to the form $\langle\ , \
\rangle_{-1}$ is well defined on $M^{N+1}$. The same is true for the
coadjoint action of $\widetilde{\mathfrak{g}}_+$ with respect to the
form $\langle\ , \ \rangle_{N}$, indeed, the factor-algebra
$\widetilde{\mathfrak{g}}_+/\sum_{l\geqslant 2N+2} \mathfrak{g}_l$
acts effectively on $M^{N+1}$.

Let $C(M^{N+1})$ be the space of smooth functions on $M^{N+1}$. For
all $f_1$, $f_2\in C(M^{N+1})$  define the \emph{first Lie-Poisson
bracket} by the formula
\begin{gather}\label{LiePoissonBra1}
  \{f_1,f_2\}_{1} = \sum_{m,n=0}^N \sum_{a,b=1}^3 P_{ab}^{mn}(-1)
  \frac{\partial f_1}{\partial \mu_m^a}\frac{\partial
  f_2}{\partial\mu_n^b},
  \intertext{where} P_{ab}^{mn}(-1) = \langle\widehat{\mu}(\lambda),
  [Z_a^{-m-1},Z_b^{-n-1}] \rangle_{-1},\notag \\[10pt]
  Z_1^{m}=H^{m},\qquad Z_2^{m}=Y^{m},\qquad Z_3^{m}=X^{m},\notag\\
  \mu_m^1=\alpha_m,\qquad \mu_m^2=\beta_m,\qquad
  \mu_m^3=\gamma_m.\notag
\end{gather}
With the same notation, define the  \emph{second Lie-Poisson
bracket} by the formula
\begin{gather}\label{LiePoissonBra2}
  \{f_1,f_2\}_2 = \sum_{m,n=0}^N \sum_{a,b=1}^3 P_{ab}^{mn}(N)
  \frac{\partial f_1}{\partial \mu_m^a}\frac{\partial
  f_2}{\partial\mu_n^b},  \intertext{where} P_{ab}^{mn}(N) = \langle\widehat{\mu}(\lambda),
   [Z_a^{-m+N},Z_b^{-n+N}] \rangle_{N}.\notag
\end{gather}

One can see that the functions $\beta_{2N+1}$ and $\gamma_{2N+1}$
\emph{annihilate the bracket (\ref{LiePoissonBra1})}
\begin{equation*}
\{\beta_{2N+1},f\}_1\,{=}\,0,\qquad \{\gamma_{2N+1},f\}_1\,{=}\,0
\qquad \text{for all} \quad f\in C(M^{N+1}).
\end{equation*}
Thus, we can assume without loss of generality that
\begin{equation}\label{MKdVconstrains}
\beta_{2N+1}=\gamma_{2N+1}=const
\end{equation}
and restrict the bracket (\ref{LiePoissonBra1}) to the subspace
$M^{N+1}_{con} \subset M^{N+1}$ with the constraints
(\ref{MKdVconstrains}), clearly, $\dim M^{N+1}_{con} = 3(N+1)$. The
first Lie-Poisson bracket is nondegenerate on $M^{N+1}_{con}$. We
use the set $\gamma_{2m-1}$, $\beta_{2m-1}$, $\alpha_{2m}$, $m=0,\,
1,\, \dots,\, N$, as \emph{coordinate functions} in $M^{N+1}_{con}$.
We call the fixed coordinates $\beta_{2N+1}$, $\gamma_{2N+1}$ the
\emph{external parameters}.

We see that, on one hand, $M^{N+1}_{con} \subset
(\widetilde{\mathfrak{g}}_-)^{\ast}$ with respect to $\langle\ , \
\rangle_{-1}$, see Example~\ref{E:k_-1}, and, on the other hand,
$M^{N+1}_{con}\subset (\widetilde{\mathfrak{g}}_+)^{\ast}$ with
respect to $\langle\ , \ \rangle_{N}$, see Example~\ref{E:k_N}.

In addition to the brackets (\ref{LiePoissonBra1}) and
(\ref{LiePoissonBra2}), one can define $N$ intermediate brackets
with the Poisson tensors
\begin{equation}\label{LiePoissonBras}
 P_{ab}^{mn}(k)=\langle\widehat{\mu}(\lambda),
   [Z_a^{-m+k},Z_b^{-n+k}] \rangle_{k}, \qquad k=0, \ldots, N-1.
\end{equation}

Now, consider the $\ad^{\ast}$-invariant function
\begin{equation*}
I(\lambda)=-\det \widetilde{\mu}(\lambda) = h_{-1} \lambda^{-1} +
h_0 + \cdots + h_{2N+1}\lambda^{2N+1}.
\end{equation*}
Then we have
\begin{equation}\label{InvFuncMKdV}
  h_{\nu}= \sum_{m+n=\nu} \left( \alpha_{2m} \alpha_{2n} +
  \gamma_{2m-1} \beta_{2n-1}\right),\qquad \nu=-1,\, 0,\,
\dots,\, 2N+1.
\end{equation}
The Kostant-Adler scheme \cite{Adler} implies the following
assertions.
\begin{proposition}
All functions $h_{\nu}$, $\nu=-1,\, 0,\, \dots,\, 2N+1$ determined
by \eqref{InvFuncMKdV} mutually commute with  respect to the
 brackets \eqref{LiePoissonBra1},
\eqref{LiePoissonBra2}, and the intermediate brackets with the
Poisson tensors \eqref{LiePoissonBras}.
\end{proposition}
\begin{proposition}
The functions $h_{\nu}$, $\nu=N, \dots,\, 2N$, are functionally
independent and annihilate the bracket \eqref{LiePoissonBra1}.

Consider the algebraic variety $\mathcal{O}^N_1\subset
M^{N+1}_{con}$ defined by the set of equations $h_{\nu}=c_{\nu},$
$\nu=N, \ldots, 2N$, where $c_{\nu}$ are arbitrary fixed complex
numbers. $\mathcal{O}^N_1$ is an orbit of coadjoint action of the
subalgebra~$\widetilde{\mathfrak{g}}_-$ and $\dim \mathcal{O}^N_1 =
2(N+1)$.
\end{proposition}
\begin{proposition}
The functions $h_{\nu}$, $\nu=-1, \dots,\, N-1$, are functionally
independent and annihilate the bracket \eqref{LiePoissonBra2}.

Consider the algebraic variety $\mathcal{O}^N_2\subset
M^{N+1}_{con}$ defined by the set of equations $h_{\nu}=c_{\nu},$
$\nu=-1, \ldots, N-1$, where $c_{\nu}$ are arbitrary fixed complex
numbers. $\mathcal{O}^N_2$ is an orbit of coadjoint action of the
subalgebra~$\widetilde{\mathfrak{g}}_+$ and $\dim \mathcal{O}^N_2 =
2(N+1)$.
\end{proposition}

It is obvious that the orbits $\mathcal{O}^N_1$ and
$\mathcal{O}^N_2$ are the \emph{symplectic leaves} with respect to
the first and the second Lie-Poisson brackets, accordingly.

Further, the functions $h_{-1}$, $h_{0}$, \ldots, $h_{N-1}$,
regarded as Hamiltonians with respect to the first Lie-Poisson
bracket, generate non-trivial flows on $M^{N+1}_{con}$
\begin{equation}\label{FlowEqMKdV}
  \frac{\partial \mu_m^a}{\partial \tau_{\nu}} =
  \{\mu_m^a,h_{\nu}\}_{1},\quad \nu=-1,\, 0,\, \dots,\, N-1.
\end{equation}
The equations (\ref{FlowEqMKdV}) can be written with the help of the
second Lie-Poisson bracket and the functions $h_{N}$, \ldots,
$h_{2N}$ regarded as Hamiltonians. Namely (see \cite{Holod85}), one
has
\begin{equation*}
  \{\mu_m^a,h_{\nu}\}_{1} = -\{\mu_m^a,h_{\nu+N+1}\}_{2}.
\end{equation*}

\begin{proposition}
The system \eqref{FlowEqMKdV} reduced to the orbit $\mathcal{O}^N_1$
is equivalent to the finite gap complex MKdV hierarchy.

The system \eqref{FlowEqMKdV} reduced to the orbit $\mathcal{O}^N_2$
is equivalent to finite gap sin(sinh)-Gordon equation.
\end{proposition}
Below we give the outline of the proof which may be found in full
detail in \cite{Bern}.

First, rewrite (\ref{FlowEqMKdV}) in matrix form
\begin{gather}
  \frac{\partial \widehat{\mu}(\lambda)}{\partial \tau_{\nu}}
  =[\nabla_2 h_{\nu+N+1}, \widehat{\mu}(\lambda)]=
  [\widehat{\mu}(\lambda),\nabla_1 h_{\nu}], \label{Laks}
  \intertext{where}  \nabla_1 h = \sum_{m=0}^N \left(
  \frac{\partial h}{\partial \alpha_{2m}}\,H^{-2m-2} +
  \frac{\partial h}{\partial \beta_{2m-1}}\,Y^{-2m-1} +
  \frac{\partial h}{\partial \gamma_{2m-1}}\,X^{-2m-1}
  \right),\notag\\
  \nabla_2 h = \sum_{m=0}^N \left(
  \frac{\partial h}{\partial \alpha_{2m}}\,H^{-2m+2N} +
  \frac{\partial h}{\partial \beta_{2m-1}}\,Y^{-2m+2N+1} +
  \frac{\partial h}{\partial \gamma_{2m-1}}\,X^{-2m+2N+1}
  \right).\notag
\end{gather}
The hamiltonian flows along $\tau_{\nu}$ and $\tau_{\nu'}$ commute,
which implies the compatibility condition in the form of \emph{zero
curvature equations}. In particular, assigning $\tau_{N-1}=x$,
$\tau_{N-2}=t$, we obtain
\begin{equation*}
  \frac{\partial \nabla_2 h_{2N}}{\partial t} -
  \frac{\partial \nabla_2 h_{2N-1}}{\partial x} + [\nabla_2 h_{2N},
  \nabla_2 h_{2N-1}]=0,
\end{equation*}
where
\begin{gather*}
  \nabla_2 h_{2N} = \begin{pmatrix} \alpha_{2N} & \beta_{2N+1} \\
  \lambda \gamma_{2N+1} & -\alpha_{2N} \end{pmatrix},  \\
  \nabla_2 h_{2N-1} = \begin{pmatrix} \alpha_{2N-2}+\lambda\alpha_{2N} &
  \beta_{2N-1}+\lambda \beta_{2N+1} \\
  \lambda \gamma_{2N-1} +\lambda^2 \gamma_{2N+1} & -(\alpha_{2N-2}+
  \lambda\alpha_2N) \end{pmatrix}.
\end{gather*}
Recall that $\beta_{2N+1}$ and $\gamma_{2N+1}$ are the fixed
external parameters.

The reduction of (\ref{FlowEqMKdV}) onto the orbit $\mathcal{O}^N_1$
gives the equation
\begin{equation}\label{MKdVEqpure}
\frac{\partial \alpha_{2N}}{\partial t} = \frac{\partial
\alpha_{2N-2}}{\partial x},
\end{equation}
equivalent to MKdV equation with respect to the function
$\alpha_{2N}(x,t)=u(x,t)$. Indeed, reducing the equation
(\ref{FlowEqMKdV}) as $\nu=N-1$ to the orbit $\mathcal{O}^N_1$  we
obtain
\begin{subequations}\label{coords}
\renewcommand{\theequation}{\theparentequation\alph{equation}}
\begin{gather}
  \beta_{2N-1}=\frac{1}{2\beta_{2N+1}} \left(c_{2N}-\frac{\partial \alpha_{2N}}
  {\partial x} - \alpha^2_{2N}\right), \\
  \gamma_{2N-1}=\frac{1}{2\beta_{2N+1}} \left(c_{2N}+\frac{\partial \alpha_{2N}}
  {\partial x} - \alpha^2_{2N}\right), \\
  \alpha_{2N-2}=\frac{1}{4\beta_{2N+1}} \left(\frac{\partial^2 \alpha_{2N}}
  {\partial x^2} - 2\alpha^3_{2N} + 2c_{2N}\alpha_{2N}\right).
  \label{alpha}
\end{gather}
\end{subequations}
It is readily seen that combining \eqref{MKdVEqpure} and
\eqref{alpha} we get the complex MKdV equation. Two real subalgebras
$\mathfrak{su}(2)$ and $\mathfrak{su}(1,1)\cong
\mathfrak{sl}(2,\Real)$ of $\mathfrak{sl}(2,\Complex)$ give rise to
two real MKdV equations (the so-called $\pm$MKdV).

In the same time, the reduction of (\ref{FlowEqMKdV}) onto the orbit
$\mathcal{O}^N_2$ leads to sin(sinh)-Gordon equation. Let
$\mathcal{O}_2^N \cup \mathfrak{g}_{-1}$ be the \emph{base} for the
orbit $\mathcal{O}^N_2$. The 1-parameter subgroup $G_0 = \exp
\mathfrak{g}_0$ parametrizes the base in a natural way
\begin{equation*}
  \gamma_{-1} = \sqrt{h_{-1}}\, e^{u},\qquad
  \beta_{-1} = \sqrt{h_{-1}}\, e^{-u}.
\end{equation*}
Then the equations (\ref{FlowEqMKdV}) imply
\begin{equation}\label{SG0}
  \alpha_{2N}
   = \frac{1}{2}\, \frac{\partial}{\partial x}\, u,
\end{equation}
where $x\equiv \tau_{N-1}$ as above; the corresponding flow is
called \emph{stationary}. The Hamiltonian $h_N$ gives rise to an
\emph{evolutionary} flow. In the case of the subalgebra
$\mathfrak{sl}(2,\Real)$ we have
\begin{equation}\label{shGordonEq}
  \frac{\partial \alpha_{2N}}{\partial t}
  = 2\beta_{2N+1}  \sqrt{h_{-1}}\, \sinh u.
\end{equation}
Combining (\ref{SG0}) and (\ref{shGordonEq}) we obtain sinh-Gordon
equation.

In the case of the subalgebra $\mathfrak{su}(2)$ we have to assign
$\alpha_{2m} = ia_{2m}$, $a_{2m}\in \Real$, and $\gamma_{2m-1} =
{-}\beta_{2m-1}^{\ast}$, therefore $\beta_{2N+1}=\gamma_{2N+1}=ib$,
$\gamma_{-1}=-ire^{iu}$, $\beta_{-1}=-ire^{-iu}$. Then we come to
sin-Gordon equation
\begin{equation*}
  \frac{\partial^2 u}{\partial t \partial x} = 4rb\sin u.
\end{equation*}

\section{Separation of variables
for MKdV and sin(sinh)-Gordon equation}\label{s:MKdVSoV}

\begin{definition}
Suppose we have the variables $(\lambda_k,w_k)$,
$k=1,\dots,N+1$, such that
\begin{enumerate}
\renewcommand{\labelenumi}{(\roman{enumi})}
\item they are quasi-canonically conjugate, that is
\begin{equation*}
  \{\lambda_k, w_l\}_1 = f(\lambda_k) \delta_{kl}, \qquad
  \{\lambda_k,\lambda_l\}_1=\{w_k,w_l\}_1=0,
\end{equation*}
where $f(\lambda)$ is an arbitrary smooth function;
\item they reduce Liouville 1-form\footnote{We call $\Omega$ Liouville 1-form if
$d\Omega=\omega$, where $\omega$ is a symplectic 2-form.} to a sum
of meromorphic differentials on the corresponding Riemannian
surface.
\end{enumerate}
We call $(\lambda_k,w_k)$, $k=1,\dots,N+1$, \emph{separation
variables}.
\end{definition}

Consider the orbit $\mathcal{O}_1^N$, $\dim \mathcal{O}_1^N =
2(N+1)$. One can parameterize the orbit using any subset of $2(N+1)$
variables from $\{\alpha_{2m},\, \beta_{2m-1}, \gamma_{2m-1}\}$,
$m=0, 1, \ldots, N$. The most natural way to obtain the
parameterization is to eliminate one of the subsets
$\{\beta_{2m-1}\}$ or $\{\gamma_{2m-1}\}$. The reason for this is
the nilpotency of the basis elements that correspond to the subsets.

Note that the correspondence between the elimination variables,
which we chose to parameterize the orbit, and the nilpotent elements
of the basis of the algebra is a crucial feature of our scheme and
applies to all examples.

We chose to parameterize the orbit $\mathcal{O}_1^N$ by the
variables $\{\gamma_{2m-1}, \alpha_{2m}\}$, $m=0,1, \ldots, N$, that
is we eliminate the set $\{\beta_{2m-1}\}$. From the orbit equations
we find
\begin{equation}\label{betaSyst}
  \beta_{2m-1} = \sum_{j=0}^{N+1} (\Gamma^{+})^{-1}_{mj} (c_{N+j}-A_{N+j}),\
  m=0,\,\dots\, N+1,\quad c_{2N+1}=\beta_{2N+1}\gamma_{2N+1},
\end{equation}
where
\begin{equation*}
   \Gamma^{+}=\begin{bmatrix} \gamma_{2N+1} & \gamma_{2N-1}
   & \dots & \gamma_1 & \gamma_{-1} \\
   0 & \gamma_{2N+1} & \dots & \gamma_3 &
   \gamma_{1} \\ \vdots & \vdots& \ddots& \vdots& \vdots \\
   0 & 0 & \dots & \gamma_{2N+1} &
   \gamma_{2N-1}\\ 0 & 0 & \dots & 0 &  \gamma_{2N+1}
   \end{bmatrix} \qquad \text{and} \qquad A_{\nu} =
   \sum_{\substack{m+n=\nu,\\ 0\leqslant m,n \leqslant N}} \alpha_{2m}\alpha_{2n}.
\end{equation*}
Now, using the parameterization (\ref{betaSyst}), we find
expressions for the Hamiltonians $h_{-1}, h_{0}$, $\ldots$,
$h_{N-1}$
\begin{equation}\label{HamiltMKdV}
   h_{n-1} = \sum_{m,j=0}^{N+1}\Gamma^{-}_{nm}(\Gamma^{+})^{-1}_{mj} (c_{N+j}-A_{N+j}) + A_{n-1},
   \quad n=0,\, \dots N,
\end{equation}
where
\begin{equation*}
   \Gamma^{-} = \begin{bmatrix} \gamma_{-1} & 0 & \dots & 0 & 0 \\
   \gamma_{1} & \gamma_{-1} & \dots & 0 & 0 \\
   \vdots & \vdots& \ddots& \vdots& \vdots \\
   \gamma_{2N-1} & \gamma_{2N-3} & \dots & \gamma_{-1}& 0
   \end{bmatrix}.
\end{equation*}
Note that the expressions (\ref{HamiltMKdV}) are linear in
$c_{\nu}$, $\nu=N, \ldots, 2N+1$.

Clearly, one can obtain an analogous parametrization of the orbit
$\mathcal{O}_2^N$ by the set $\{\alpha_{2m},\gamma_{2m-1}\}$,
$m=0,\ldots, N$.

To proceed we need to define the \emph{characteristic polynomial}
\begin{equation*}
  Q(\varkappa,\lambda) = \det\bigl(\mu(\lambda)-\varkappa \cdot I\bigr),
\end{equation*}
where $I$ denotes $2 \times 2$ identity matrix. By the substitution
$\varkappa=w \lambda^{-1}$ the equation $Q(\varkappa,\lambda)=0$
becomes transformed into the standard equation of a hyperelliptic
curve of genus $N+1$
\begin{equation}\label{CharPolyMKdV}
  P(w,\lambda)=\lambda^2Q(w\lambda^{-1},\lambda)=w^2 - \lambda(h_{-1}+h_0\lambda + \cdots+ h_{2N+1}
  \lambda^{2N+2})=0.
\end{equation}

Recall that on the orbit $\mathcal{O}_1^N$ we have
$h_{\nu}=c_{\nu}$, $\nu=N$, \ldots, $2N$. Denote by $(w_k,\,
\lambda_k)$ a root of $P(w,\lambda)$ on the orbit, that is
\begin{equation}\label{HyperCurveMKdV}
  w_k^2 = \lambda_k (h_{-1}+h_0\lambda_k + \cdots h_{N-1}\lambda_k^{N} +
  c_N \lambda_k^{N+1} + c_{N+1} \lambda_k^{N+2} + \cdots + c_{2N+1} \lambda_k^{2N+2}).
\end{equation}
We proceed to show that the set $\{(w_k,\, \lambda_k)\}$, $k=1,
\ldots, N+1$, defines another parametrization of the orbit
$\mathcal{O}_1^N$. We have to find the explicit relation between the
sets $\{(w_1,\lambda_1),\ldots,$ $(w_{N+1},\lambda_{N+1})\}$ and
$\{\alpha_{0},\alpha_{2},\ldots,\alpha_{2N},
\gamma_{-1},\gamma_{1},\ldots,\gamma_{2N-1}\}$.

Solving (\ref{HyperCurveMKdV}) for the Hamiltonians $h_{-1}, h_0,
\ldots, h_{N-1}$ one gets
\begin{equation}\label{hSyst}
  \begin{array}{l} h_{-1} = \displaystyle\frac{1}{W}
  [W_1\left(\textstyle\frac{w^2}{\lambda}\right)
  -c_N W_1(\lambda^{N+1}) - \cdots - c_{2N+1}
  W_1(\lambda^{2N+2})]\\
  h_{0} = \displaystyle\frac{1}{W}
  [W_2\left(\textstyle\frac{w^2}{\lambda}\right)
  -c_N W_2(\lambda^{N+1}) - \cdots - c_{2N+1}
  W_2(\lambda^{2N+2})]\\ .\ .\ .\ .\ .\ .\ .\ .\ .\ .\ .\ .\ .\ .\ .\ .\
  .\ .\ .\ .\ .\ .\ .\ .\ .\ .\ .\ .\ .\ .\ .\ .\ .\ .\ .\ .\ .\ .\ .\ .\ \\
  h_{N-1} = \displaystyle\frac{1}{W}
  [W_{N+1}\left(\textstyle\frac{w^2}{\lambda}\right)
  -c_N W_{N+1}(\lambda^{N+1}) - \cdots - c_{2N+1}
  W_{N+1}(\lambda^{2N+2})].
  \end{array}
\end{equation}
where $W = \prod(\lambda_i-\lambda_j)$ is Vandermonde determinant of
$\lambda_1$, $\lambda_2$, \ldots, $\lambda_{N+1}$. By
$W_i(f(\lambda,w))$ we denote the determinant of Vandermonde matrix
with the $i$-th column replaced by  $\bigl(f(\lambda_1, w_1), \dots,
f(\lambda_{N+1}, w_{N+1})\bigr)^t$.

On the orbit the formulas (\ref{HamiltMKdV}) and (\ref{hSyst})
define the same set of functions. We see that  both
(\ref{HamiltMKdV}) and (\ref{hSyst}) are linear in $c_{\nu}$,
$\nu=N, \ldots, 2N+1$. As $\{c_{\nu}\}$ is the set of
\emph{independent} parameters one can equate the corresponding
terms. Namely, we get
\begin{equation*}
  \frac{\gamma_{2m-1}}{\gamma_{2N+1}} =
  -\frac{W_{m+1}(\lambda^{N+1})}{W}, \quad m=0,\,\dots N.
\end{equation*}
This implies that the set $\{\lambda_k\}$ is, in fact, the set of
roots of the polynomial $\gamma(\lambda)$
\begin{equation*}
   \gamma(\lambda_k)=0,
\end{equation*}
while the variables $\{w_k\}$ satisfy the equalities
\begin{equation*}
w_k^2 = \lambda_k^2
\left(\alpha^2(\lambda_k)-\gamma(\lambda_k)\beta(\lambda_k)\right) =
\lambda_k^2 \alpha^2(\lambda_k),\qquad k=1,\dots,N+1.
\end{equation*}

\begin{theorem}\label{T:SVMKdV}
Suppose the orbit $\mathcal{O}_1^N$ has the coordinates
$(\alpha_{2m},\, \gamma_{2m-1})$, $m=0, 1, \ldots, N$, as above.
Then the new coordinates $(\lambda_k, w_k)$, $k=1,\dots,N+1$,
defined by the formulas
\begin{equation}
  \gamma(\lambda_k)=0,\qquad
  w_k = \varepsilon \lambda_k \alpha(\lambda_k),
  \qquad \text{where} \quad \varepsilon^2=1,\label{newvar1}
\end{equation}
have the following properties:
\begin{enumerate}
\item[\textup{(1)}]  a pair $(w_k,\lambda_k)$ is a root of the characteristic
polynomial \eqref{CharPolyMKdV}.
\item[\textup{(2)}] a pair
$(\lambda_k, w_k)$ is quasi-canonically conjugate with respect to
the first Lie-Poisson bracket \eqref{LiePoissonBra1}:
\begin{equation}\label{PoissonBra1}
  \{\lambda_k,\lambda_l\}_1=0, \qquad
  \{\lambda_k, w_l\}_1 = \varepsilon\lambda_k \delta_{kl}, \qquad \{w_k,w_l\}_1=0;
\end{equation}
\item[\textup{(3)}] the corresponding  Liouville 1-form is
\begin{align*}
&\Omega_{-1}=\sum\limits_{k}\varepsilon\lambda_{k}^{-1}w_{k}\,d\lambda_{k}.&
\end{align*}
\end{enumerate}
\end{theorem}

\begin{proof}

(1) The assertion is a direct consequence of \eqref{CharPolyMKdV}
and \eqref{newvar1}.

(2) It is evident that
\begin{equation*}
\{\lambda_k,\lambda_l\}_1=0.
\end{equation*}
Indeed, since $\lambda_k$, $k=1$, \ldots, $N+1$, depend only on
$\gamma_{2m-1}$, $m=0$, \ldots, $N$, and $\gamma_{2m-1}$ mutually
commute, $\lambda_k$ also mutually commute.

Let us calculate the bracket of $\lambda_k$ and $w_l$
\begin{equation*}
  \{\lambda_k, w_l\}_1 = \sum_{m,n}
  \left(\frac{\partial \lambda_k}{\partial \gamma_{2m-1}}
        \frac{\partial w_l}{\partial \alpha_{2n}} -
        \frac{\partial \lambda_k}{\partial \alpha_{2n}}
        \frac{\partial w_l}{\partial \gamma_{2m-1}} \right)
        \{\gamma_{2m-1}, \alpha_{2n}\}_1.
\end{equation*}
From \eqref{newvar1} we have
\begin{equation}\label{lambdawAdd}
\frac{\partial \lambda_k}{\partial \alpha_{2n}}=0,\qquad
  \frac{\partial \lambda_k}{\partial \gamma_{2m-1}}=
  -\frac{\lambda_k^m}{\gamma'(\lambda_k)},\qquad
  \frac{\partial w_l}{\partial \alpha_{2n}} = \varepsilon\lambda_l^{n+1}.
\end{equation}
Further  $\{\gamma_{2m-1}, \alpha_{2n}\}_1 = -\gamma_{2(m+n)+1}$
when $m+n<N$ and $\{\gamma_{2m-1}, \alpha_{2n}\}_1 = 0$ when
$m+n\geqslant N$. Thus, we obtain
\begin{equation*}
  \{\lambda_k, w_l\}_1 = \frac{ \sum\limits_{m+n<N}
  \varepsilon \lambda_k^m \lambda_l^{n+1}\gamma_{2(m+n)+1}}{\gamma'(\lambda_k)}  =
   \frac{\varepsilon \lambda_l}{\gamma'(\lambda_k)}
   \frac{\gamma(\lambda_l)-\gamma(\lambda_k)}{\lambda_l-\lambda_k}.
\end{equation*}
As $k\neq l$ it is evident that $\{\lambda_k, w_l\}_1=0$ while
$\gamma(\lambda_l)=\gamma(\lambda_k)=0$. As $k=l$  we get
\begin{equation*}
\{\lambda_k, w_k\}_1 = \lim_{\lambda_l\to\lambda_k}
\frac{\varepsilon\lambda_l}{\gamma'(\lambda_k)}
   \frac{\gamma(\lambda_l)-\gamma(\lambda_k)}{\lambda_l-\lambda_k}=
 \varepsilon\lambda_k.
\end{equation*}
Thus,
\begin{equation*}
  \{\lambda_k, w_l\}_1 = \varepsilon\lambda_k \delta_{kl}.
\end{equation*}

Let us calculate the bracket of $w_k$ and $w_l$
\begin{equation*}
  \{w_k, w_l\}_1 = \sum_{m,n}
  \left(\frac{\partial w_k}{\partial \gamma_{2m-1}}
        \frac{\partial w_l}{\partial \alpha_{2n}} -
        \frac{\partial w_k}{\partial \alpha_{2n}}
        \frac{\partial w_l}{\partial \gamma_{2m-1}} \right)
        \{\gamma_{2m-1}, \alpha_{2n}\}_1.
\end{equation*}
From \eqref{newvar1} follows that
$$\frac{\partial w_k}{\partial \gamma_{2m-1}} =
\varepsilon\left[\alpha(\lambda_k)+\lambda_k
\alpha'(\lambda_k)\right]\frac{\partial \lambda_k}{\partial
\gamma_{2m-1}},$$  then, using (\ref{lambdawAdd}), we obtain
\begin{equation*}
  \{w_k, w_l\}_1 = \left(\frac{\lambda_l \left[\alpha(\lambda_k)+\lambda_k
\alpha'(\lambda_k)\right]}{\gamma'(\lambda_k)}
  -  \frac{\lambda_k \left[\alpha(\lambda_l)+\lambda_l
\alpha'(\lambda_l)\right]}{\gamma'(\lambda_l)} \right)
   \frac{\gamma(\lambda_l)-\gamma(\lambda_k)}{\lambda_l-\lambda_k},
\end{equation*}
hence  $$\{w_k,w_l\}_1=0.$$

(3) From (\ref{PoissonBra1}) it follows that Liouville 1-form on the
orbit $\mathcal{O}_{1}^{N}$ is
\begin{equation*}
\Omega_{-1}=\sum\limits_{k}\varepsilon
\lambda_{k}^{-1}w_{k}\,d\lambda_{k}.
\end{equation*}

The reduction to Liouville torus is done by fixing the values of
Hamiltonians $h_{-1}$, $h_0$, \ldots, $h_{N-1}$. On the torus $w_k$
is the algebraic function of $\lambda_k$ due to(\ref{CharPolyMKdV}).
After the reduction the form $\Omega_{-1}$ becomes a sum of
meromorphic differentials on the Riemann surface $P(w,\lambda)=0$.
\end{proof}

The next theorem is proven similarly.
\begin{theorem}\label{T:SVsinG}
Suppose the orbit $\mathcal{O}_2^N$ has the coordinates
$(\alpha_{2m},\, \gamma_{2m-1})$, $m=0, 1, \ldots, N$. Then the new
coordinates $(\lambda_k, w_k)$, $k=1,\ldots, N+1$, defined by the
formulas
\begin{equation*}
  \gamma(\lambda_k)=0,\qquad
  w_k = \varepsilon \lambda_k \alpha(\lambda_k),
  \qquad \text{where} \qquad \varepsilon^2=1,
\end{equation*}
have the following properties:
\begin{enumerate}
\item[\textup{(1)}] a pair $(w_k,\lambda_k)$ is a root of the characteristic
polynomial \eqref{CharPolyMKdV};
\item[\textup{(2)}] a pair
$(\lambda_k, w_k)$ is quasi-canonically conjugate with respect to
the second Lie-Poisson bracket \eqref{LiePoissonBra2}:
\begin{equation}\label{PoissonBra2}
  \{\lambda_k,\lambda_l\}_2=0, \qquad
  \{\lambda_k, w_l\}_2 = - \varepsilon\lambda_k^{N+2} \delta_{kl}, \qquad \{w_k,w_l\}_2=0;
\end{equation}
\item[\textup{(3)}] the corresponding  Liouville 1-form is
\begin{equation*}
  \Omega_{N}=-\sum\limits_{k}
  \varepsilon\lambda_{k}^{-(N+2)}w_{k}\,d\lambda_{k}.
\end{equation*}
\end{enumerate}
\end{theorem}

Let us summarize our scheme of obtaining the separation variables.
First, we parameterize the orbit by eliminating a subset of group
coordinates corresponding to nilpotent basis elements. Next, we
restrict the curve $P(w,\lambda)\,{=}\,0$ onto the orbit, where
$P(w,\lambda)$ is the characteristic polynomial
\begin{equation*}
  P(w,\lambda)=\det \bigl(\mu(\lambda)-w\cdot I\bigr),
\end{equation*}
$I$ is identity matrix. We use the set $\{\lambda_k,w_k\}$,
$k=1,\dots, N+1$, where $P(w_k,\lambda_k)=0$, to define another
parametrization of the orbit. Then, we equate expressions for
Hamiltonians in the coordinates of two parameterizations of the
orbit in order to obtain the link between the two sets of orbit
coordinates. Finally, the set $\{\lambda_k,w_k\}$ is the set of
separation variables.

Further, in Sections 4 and 5 we apply the scheme to nonlinear
Schr\"{o}dinger equation and Heisenberg magnetic chain.

\section{Phase spaces for nonlinear
Schr\"{o}dinger equation and \\ Heisenberg magnetic chain as orbits
in $\mathfrak{sl}(2,\Complex)\otimes \mathcal{P}(z,z^{-1})$}

Here we use the construction from Section 1 with \emph{homogeneous
grading}. That is,
\begin{equation}\label{basisNLS&HM}
  X^l = z^l X,\qquad Y^l = z^l Y, \qquad H^l = z^l H
\end{equation}
be the \emph{basis} in $\widetilde{\mathfrak{g}}\simeq
\mathfrak{sl}(2, \mathbb{C})\bigotimes\mathcal{P}(z, z^{-1})$.

Note the well-known fact that the Lie algebra from Sections 1--2 can
be realized as the subalgebra of $\mathfrak{sl}(2,\Complex)\otimes
\mathcal{P}(z,z^{-1})$ invariant with respect to an automorphism of
order 2, see \cite{Holod95}, \cite{Kac}, \cite{Kroode}.

By $\mathfrak{g}_l$, $l\in \Integer$, denote an eigenspace of
homogeneous degree $l$. It is evident that
\begin{equation*}
   \mathfrak{g}_l = \spanOp_{\Complex} \{X^l,\, Y^l,\, H^l\}.
\end{equation*}

Decompose $\widetilde{\mathfrak{g}}$ into two subalgebras
\begin{equation*}
  \widetilde{\mathfrak{g}}_+ = \sum_{l\geqslant
  0}\mathfrak{g}_l,\qquad
  \widetilde{\mathfrak{g}}_- = \sum_{l < 0}\mathfrak{g}_l,\qquad
  \widetilde{\mathfrak{g}} = \widetilde{\mathfrak{g}}_+ +
  \widetilde{\mathfrak{g}}_-.
\end{equation*}
Use the same \emph{$\ad$-invariant bilinear forms}
\eqref{bilinForms} to define the spaces dual to
$\widetilde{\mathfrak{g}}_+$ and $\widetilde{\mathfrak{g}}_-$.

Fix $N\geqslant 0$. Consider $M^{N+1}\subset
\widetilde{\mathfrak{g}}$, where an element $\widehat{\mu}(z)\in
M^{N+1}$ has the form
\begin{equation*}
  \widehat{\mu}(z) =  \begin{pmatrix}
  \alpha(z) & \beta(z) \\ \gamma(z) &
  -\alpha(z) \end{pmatrix}
\end{equation*}
with
\begin{equation*}
  \alpha(\lambda)=\sum_{m=0}^{N+1} z^m
  \alpha_{m},\qquad \beta(z) = \sum_{m=0}^{N+1} z^{m}
  \beta_{m},\qquad \gamma(z) = \sum_{m=0}^{N+1}z^{m}
  \gamma_{m}.
\end{equation*}
As above, we call $M^{N+1}$ the \emph{$N$-gap sector of
$\widetilde{\mathfrak{g}}$}, or shortly the \emph{finite gap
sector}.

For all $f_1, f_2 \in C(M^{N+1})$ define two Lie-Poisson brackets
\begin{gather}\label{LiePoissonBraNLS1}
  \{f_1,f_2\}_{1} = \sum_{m,n=0}^{N+1} \sum_{a,b=1}^3 P_{ab}^{mn}(-1)
  \frac{\partial f_1}{\partial \mu_m^a}\frac{\partial
  f_2}{\partial\mu_n^b} \intertext{and} \label{LiePoissonBraNLS2}
  \{f_1,f_2\}_{2} = \sum_{m,n=0}^{N+1} \sum_{a,b=1}^3 P_{ab}^{mn}(N+1)
  \frac{\partial f_1}{\partial \mu_m^a}\frac{\partial f_2}{\partial\mu_n^b},
  \intertext{where} P_{ab}^{mn}(-1) = \langle\widehat{\mu}(z),
  [Z_a^{-m-1},Z_b^{-n-1}] \rangle_{-1}, \notag \\
  P^{mn}_{ab}(N+1) = \langle\widehat{\mu}(z),
  [Z_a^{-m+N+1},Z_b^{-n+N+1}] \rangle_{N+1}, \notag \\[10pt]
  Z_1^{m}=H^{m},\qquad Z_2^{m}=Y^{m},\qquad Z_3^{m}=X^{m},\notag\\
  \mu_m^1=\alpha_m,\qquad \mu_m^2=\beta_m,\qquad
  \mu_m^3=\gamma_m.\notag
\end{gather}

One can see that $M^{N+1}\subset
(\widetilde{\mathfrak{g}}_{-})^{\ast}$ with respect to $\langle\ ,\
\rangle_{-1}$ and, in the same time, $M^{N+1}\subset
(\widetilde{\mathfrak{g}}_{+})^{\ast}$ with respect to $\langle\ ,\
\rangle_{N+1}$.

Next, introduce the $\ad^{\ast}$-invariant function
\begin{equation*}
  I(z)=-\det \widetilde{\mu}(z) = h_0 + h_1 z + \cdots +
  h_{2N+2}z^{2N+2},
\end{equation*}
where
\begin{equation}\label{InvFuncNLS}
  h_{\nu}= \sum_{m+n=\nu} \left( \alpha_{m} \alpha_{n} +
  \gamma_{m} \beta_{n}\right),\qquad \nu=0,1,\dots,2N+2.
\end{equation}

One can easily prove that the functions $\alpha_{N+1}$,
$\beta_{N+1}$, $\gamma_{N+1}$ annihilate the bracket
\eqref{LiePoissonBraNLS1}.

In order to obtain nonlinear Schr\"{o}dinger equation we have to
assign
\begin{equation}\label{NLSconstraines}
    \beta_{N+1}=\gamma_{N+1}=0,\qquad \alpha_{N+1}=const\neq 0.
\end{equation}
After the restriction of the bracket \eqref{LiePoissonBraNLS1} to
the subspace $M_{con}^{N+1}\subset M^{N+1}$ with the constrains
\eqref{NLSconstraines} we get $\dim M_{con}^{N+1} = 3(N+1)$. The
bracket \eqref{LiePoissonBraNLS1} is nondegenerate on
$M_{con}^{N+1}$. We use the set $\gamma_m, \beta_m, \alpha_m$,
$m=0,1,\dots N$ as \emph{coordinate functions} in $M_{con}^{N+1}$.
We call the fixed coordinate $\alpha_{N+1}$ the \emph{external
parameter}.

On the other hand, the functions $\alpha_{N+1}$, $\beta_{N+1}$,
$\gamma_{N+1}$ commute with all Hamiltonians $h_{\nu}$, $\nu=N+2,
N+3, \dots, 2N+2$, with respect to the bracket
\eqref{LiePoissonBraNLS2} and give rise to nontrivial flows on
$M^{N+1}$, therefore are Hamiltonians. That is, the bracket
\eqref{LiePoissonBraNLS2} is considered on $M^{N+1}$, $\dim M^{N+1}
= 3(N+2)$, and the set $\gamma_m$, $\beta_m$, $\alpha_m$,
$m=0,1,\dots N+1$ serve as \emph{coordinate functions} in $M^{N+1}$.

The following assertions are immediately derived from the
Kostant-Adler scheme \cite{Adler}.
\begin{proposition}
All functions $h_{\nu}$, $\nu=0,1,\dots,2N+2$ determined by
\eqref{InvFuncNLS} mutually commute with respect to the brackets
\eqref{LiePoissonBraNLS1} and \eqref{LiePoissonBraNLS2}.
\end{proposition}
\begin{proposition}
The functions $h_{\nu}$, $\nu=N+1,\dots, 2N+1$ are functionally
independent and annihilate the bracket \eqref{LiePoissonBraNLS1}.

Consider the algebraic variety $\mathcal{O}_1^{N}\subset
M_{con}^{N+1}$ defined by the set of equation $h_{\nu}=c_{\nu}$,
$\nu=N+1,\dots,2N+1$, where $c_{\nu}$ are arbitrary fixed complex
numbers. $\mathcal{O}_1^{N}$ is an orbit of coadjoint action of the
subalgebra $\widetilde{\mathfrak{g}}_{-}$ and $\dim
\mathcal{O}_1^{N} = 2(N+1)$.
\end{proposition}
\begin{proposition}
The functions $h_{\nu}$, $\nu=0,\dots, N+1$ are functionally
independent and annihilate the bracket \eqref{LiePoissonBraNLS2}.

Consider the algebraic variety $\mathcal{O}_2^{N+1}\subset M^{N+1}$
defined by the set of equation $h_{\nu}=c_{\nu}$, $\nu=0,\dots,N+1$,
where $c_{\nu}$ are arbitrary fixed complex numbers.
$\mathcal{O}_2^{N+1}$ is an orbit of coadjoint action of the
subalgebra $\widetilde{\mathfrak{g}}_{+}$ and $\dim
\mathcal{O}_2^{N+1} = 2(N+2)$.
\end{proposition}

Further, the functions $h_0, h_1, \ldots, h_N$ regarded as
Hamiltonians with respect to the bracket \eqref{LiePoissonBraNLS1}
give rise to nontrivial flows on $M_{con}^{N+1}$
\begin{equation}\label{FlowEqNLS}
  \frac{\partial \mu_m^a}{\partial
  \tau_{\nu}}=\{\mu_m^a,h_{\nu}\}_1,\qquad \nu=0,1,\dots N.
\end{equation}
 The functions
$h_{N+2}, \dots, h_{2N+2}$ regarded as Hamiltonians with respect to
the bracket \eqref{LiePoissonBraNLS2} give rise to nontrivial flows
on $M^{N+1}$
\begin{equation}\label{FlowEqHM}
  \frac{\partial \mu_m^a}{\partial
  \tau_{\nu}}=-\{\mu_m^a,h_{\nu+N+2}\}_2,\qquad \nu=0,1,\dots N.
\end{equation}

\begin{proposition}
The system \eqref{FlowEqNLS} reduced to the orbit $\mathcal{O}_1^N$
is equivalent to finite gap nonlinear Schr\"{o}dinger equation.

The system \eqref{FlowEqHM} reduced to the orbit
$\mathcal{O}_2^{N+1}$ is equivalent to finite gap Heisenberg
magnetic chain.
\end{proposition}
Here we give the outline of the proof.

Consider the orbit $\mathcal{O}_1^N$. Rewrite \eqref{FlowEqNLS} in
matrix form. In particular, assigning $\tau_{N}=x$, $\tau_{N-1}=t$,
we obtain
\begin{subequations}\label{FlowNLSEq}
\begin{gather}
  \frac{\partial \widehat{\mu}(z)}{\partial x}
  =[\widehat{\mu}(z),\nabla_1 h_{N}]=[\nabla_2 h_{2N+1},\widehat{\mu}(z)],\\
  \frac{\partial \widehat{\mu}(z)}{\partial \tau}
  =[\widehat{\mu}(z),\nabla_1 h_{N-1}]=[\nabla_2 h_{2N},\widehat{\mu}(z)], \intertext{where}
  \nabla_2 h_{2N+1} = \begin{pmatrix} z\alpha_{2N+1}+\alpha_{N} & \beta_{N} \\
  \gamma_{N} & -(z\alpha_{2N+1}+\alpha_{N})
  \end{pmatrix},\notag \\
  \nabla_2 h_{2N} = \begin{pmatrix} z^2\alpha_{2N+1}+
  z\alpha_{N}+\alpha_{N-1} &
  z\beta_{N}+\beta_{N-1} \\
  z \gamma_{N} +\gamma_{N-1} & -(z^2\alpha_{N+1}+
  z\alpha_{N}+\alpha_{N-1}) \end{pmatrix}. \notag
\end{gather}
\end{subequations}

We use the real subalgebra $\mathfrak{su}(2)$ of
$\mathfrak{sl}(2,\Complex)$, that is assign $\alpha_{m}=ia_{m}$,
$\gamma_{m}=\mp\beta_{m}^{*}$, then the compatibility condition for
(\ref{FlowNLSEq}) gives
\begin{equation*}
 2i a_{N+1}\frac{\partial \beta_N}{\partial t}=
  -\frac{\partial^2 \beta_N}{\partial x^2}
   - 2\beta_N|\beta_N|^2-2\beta_N h_{2N},
\end{equation*}
which coincides with nonlinear Schr\"{o}dinger equation with respect
to the function $\beta_{N}(x,t)=\psi(x,t)$ as $a_{N+1}=\frac{1}{2}$,
$h_N=a_{N}=0$,
\begin{equation*}
 i\frac{\partial \psi}{\partial t}= -\frac{\partial^2 \psi}{\partial
 x^2}+2\varepsilon\psi|\psi|^{2},\qquad \varepsilon^2=1.
\end{equation*}

Consider the orbit $\mathcal{O}_2^{N+1}$. Note, that $\dim
\mathcal{O}_2^{N+1} = 2(N+2)$ and the set $h_{N+2}$, $h_{N+3}$,
\ldots, $h_{2N+2}$ is insufficient to provide Liouville
integrability. Since the functions $\alpha_{N+1}$, $\beta_{N+1}$,
and $\gamma_{N+1}$ are in involution with the set $h_{N+2}$,
$h_{N+3}$, \ldots, $h_{2N+2}$ with respect to the bracket
\eqref{LiePoissonBraNLS2} one can take any of them as an extra
Hamiltonian. Here we chose $\alpha_{N+1}$.

Rewrite \eqref{FlowEqHM} in matrix form. In particular, by assigning
$\tau_{N+2}=x$ and $\tau_{N+3}=t$ we obtain
\begin{subequations}\label{FlowHMEq}
\begin{gather}
  \frac{\partial \widehat{\mu}(z)}{\partial x}=
  [\nabla_2 h_{N+2},\widehat{\mu}(z)]=[\widehat{\mu}(z),\nabla_1 h_0], \\
   \frac{\partial \widehat{\mu}(z)}{\partial t}=
   [\nabla_2 h_{N+3},\widehat{\mu}(z)]=[\widehat{\mu}(z),\nabla_1 h_1], \intertext{where}
  \nabla_1 h_0 = z^{-1}\begin{pmatrix} \alpha_0 & \beta_0 \\
  \gamma_0 & -\alpha_0 \end{pmatrix},\notag \\
  \nabla_1 h_1 = z^{-1}\begin{pmatrix} \alpha_1 & \beta_1 \\
  \gamma_1 & -\alpha_1 \end{pmatrix} +  z^{-2}\begin{pmatrix} \alpha_0 & \beta_0 \\
  \gamma_1 & -\alpha_1 \end{pmatrix}. \notag
\end{gather}
\end{subequations}

Let us replace the coordinates $\alpha_m$, $\beta_m$, $\gamma_m$,
$m=0, 1, \ldots, N+1$, according to the formulas
\begin{gather*}
  \alpha_m=i\mu_m^3,\qquad \beta_m=\mu_m^1-i\mu_m^2,\qquad \gamma_m =
  -\mu_m^1-i\mu_m^2.
\end{gather*}
Then
\begin{equation*}
\{\mu_m^i,\mu_n^j\}_2 = \varepsilon_{ijk} \mu_{m+n-N-1}^{k}.
\end{equation*}

Introduce the vector notation
$\bm{\mu}_m=(\mu_m^1,\mu^2_m,\mu^3_m)^t$, $m=0, 1, \ldots, N+1$. Now
the orbit $\mathcal{O}_2^{N+1}$ is determined by the equations
\begin{gather*}
  \begin{array}{l}
  (\bm{\mu}_{0},\bm{\mu}_{0})=-c_{0},\\
  2(\bm{\mu}_{0},\bm{\mu}_{1})=-c_{1}, \\
  .\ .\ .\ .\ .\ .\ .\ .\ .\ .\ .\ .\ .\ .\ .\ .\ \\
  \sum\limits_{m+n=N+1}(\bm{\mu}_{m},\bm{\mu}_{n})=-c_{N+1},
  \end{array}
 \end{gather*}
where $\left(\cdot, \cdot \right)$ denotes the dot product. The
equations (\ref{FlowHMEq}) are written in the form
\begin{subequations}\label{FlowHMEq1}
\renewcommand{\theequation}{\theparentequation\alph{equation}}
\begin{gather}
  \frac{\partial \bm{\mu}_m}{\partial x} = 2[\bm{\mu}_0,
  \bm{\mu}_{m+1}],  \\
  \frac{\partial \bm{\mu}_m}{\partial t} = 2[\bm{\mu}_1,
  \bm{\mu}_{m+1}] + 2[\bm{\mu}_0,\bm{\mu}_{m+2}],
\end{gather}
\end{subequations}
where $\left[\cdot, \cdot \right]$ denotes the cross product.

By reduction of (\ref{FlowHMEq1}) to the orbit $\mathcal{O}_2^{N+1}$
we obtain
\begin{equation*}
  \bm{\mu}_{1}=\frac{1}{2c_{0}}\left[\bm{\mu}_{0},\frac{\partial \bm{\mu}_{0}}{\partial
  x}\right]+\frac{c_1}{2c_0}\bm{\mu}_0.
\end{equation*}
Taking into account the compatibility condition
\begin{equation*}
    \frac{\partial\bm{\mu}_{0}}{\partial t}=
    \frac{\partial\bm{\mu}_{1}}{\partial x},
\end{equation*}
we get
\begin{equation}\label{HMEq}
    \frac{\partial\bm{\mu}_{0}}{\partial t}=\frac{1}{2c_0}
    \left[\bm{\mu}_{0}, \frac{\partial^{2}\bm{\mu}_{0}}{\partial
    x^2}\right]+\frac{c_1}{2c_0}\frac{\partial\bm{\mu}_{0}}{\partial
    x}.
\end{equation}
When $c_{1}=0$, the equation (\ref{HMEq}) becomes the well-known
\emph{classic Heisenberg magnetic equation}, also called
\emph{isotropic Landau-Livshits equation}.

\section{Separation of variables for  nonlinear Schr\"{o}dinger
 equation}
Consider the orbit $\mathcal{O}_1^N$, $\dim \mathcal{O}_1^N=2(N+1)$.
The most natural way to parameterize the orbit, which we already
noted in Section 2, is to eliminate the subset $\{\beta_m\}$ (or
$\{\gamma_m\}$), $m=0,1,\dots, N+1$. Then, roots of
$\gamma(\lambda)$ (or $\beta(\lambda)$) give a half of separation
variables. This way is applied in \cite{Its}. However, in the case
of nonlinear Schr\"{o}dinger equation the finite gap phase space is
determined by the constrains \eqref{NLSconstraines}. That is, the
polynomials $\beta(\lambda)$ and $\gamma(\lambda)$ have the order
$N$, therefore the set of roots is insufficient to parameterize the
orbit~$\mathcal{O}_1^N$.

In order to solve this problem we change coordinates. We use
coordinates as in \cite{Previato}. Let
\begin{equation}\label{NLSbasis}
  T=\begin{pmatrix} 0&-\frac{1}{2} \\ -\frac{1}{2} &0
  \end{pmatrix},\qquad
  R=\begin{pmatrix} \frac{1}{2}& -\frac{1}{2} \\
  \frac{1}{2}&-\frac{1}{2}\end{pmatrix},\qquad
  S=\begin{pmatrix} \frac{1}{2}& \frac{1}{2} \\
  -\frac{1}{2}&-\frac{1}{2}\end{pmatrix},
\end{equation}
be the basis in $\mathfrak{sl}(2,\Complex)$. It is easily shown that
$[T,S]=S$, $[T,R]=-R$, $[S,R]=2T$. Then
\begin{equation*}
  T^{m} = z^m T,\qquad R^{m}=z^m R,\qquad S^{m}=z^{m} S.
\end{equation*}
is a \emph{basis} in $\widetilde{\mathfrak{g}}\simeq
\mathfrak{sl}(2, \mathbb{C})\otimes \mathcal{P}(z,z^{-1})$. An
element $\widehat{\mu}(z)\in M^{N+1}$ has the form
\begin{equation*}
 \widehat{\mu}(z)=\begin{pmatrix} \frac{1}{2}[r(z)+s(z)] &
 \frac{1}{2}[r(z)-s(z)-2t(z)]
  \\  \frac{1}{2}[s(z)-r(z)-2t(z)] &
  -\frac{1}{2}[r(z)+s(z)]  \end{pmatrix},
\end{equation*}
where
\begin{equation*}
  t(z)=\sum_{m=0}^{N+1} z^m t_m,\qquad r(z)=\sum_{m=0}^{N+1} z^m r_m,\qquad
  s(z)=\sum_{m=0}^{N+1} z^m s_m.
\end{equation*}
Note that $t_{m}$, $r_{m}$, $s_{m}$, $m=0,1,\dots, N+1$ are defined
by the formulas:
\begin{align}
&t_{m}=\langle\widehat{\mu}(z),
T^{-m-1}\rangle_{-1}=\langle\widehat{\mu}(z),
T^{-m+N+1}\rangle_{N+1},\notag  \\
&r_{m}=\langle\widehat{\mu}(z),
R^{-m-1}\rangle_{-1}=\langle\widehat{\mu}(z),
R^{-m+N+1}\rangle_{N+1},\label{PreviatoCoord} \\
&s_{m}= \langle\widehat{\mu}(z),
S^{-m-1}\rangle_{-1}=\langle\widehat{\mu}(z),
S^{-m+N+1}\rangle_{N+1}, \qquad m=0,1,\dots, N+1. \notag
\end{align}

With the new coordinates $M_{con}^{N+1}$ is defined by constrains
\begin{equation*}
  t_{N+1}=0,\qquad r_{N+1}=s_{N+1}=\sqrt{c_{2N+2}}\neq 0.
\end{equation*}

One can see that the subsets $\{r_m\}$ and $\{s_m\}$, $m=0,1,\dots,
N+1$ have nilpotent corresponding basis elements. For the reason we
chose one of these subsets, namely $\{r_m\}$ here, to parameterize
the orbit $\mathcal{O}_1^N$. From the orbit equations we find
\begin{equation}\label{rSyst}
  r_{m} = \sum_{j=0}^{N+1} (S^{+})_{mj}^{-1}(c_{j+N+1}-B_{j+N+1}),
  \ m=0,\dots, N+1,\quad c_{2N+2}=r_{N+1}s_{N+1},
\end{equation}
where
\begin{equation*}
  S^{+} = \begin{bmatrix} s_{N+1} & s_{N}  & \dots & s_1 & s_{0} \\
  0 & s_{N+1} & \dots & s_2 &  s_{1} \\  \vdots & \vdots& \ddots& \vdots& \vdots \\
  0 & 0 & \dots & s_{N+1} & s_{N}\\ 0 & 0 & \dots & 0 &  s_{N+1}
  \end{bmatrix}\qquad \text{and} \qquad
  B_{\nu} =\sum_{\substack{m+n=\nu,\\ 0\leqslant m,n \leqslant N}} t_m t_n  .
\end{equation*}
Using the parameterization \eqref{rSyst}, we find the expressions
for the Hamiltonians $h_0, h_1, \dots, h_N$
\begin{equation}\label{HamiltNLS}
   h_{n} = \sum_{mj=0}^{N+1} S^{-}_{nm}(S^{+})_{mj}^{-1}(c_{j+N+1}-B_{j+N+1}) + B_n,\qquad
   n=0,\,\dots N,
\end{equation}
where
\begin{equation*}
S^{-} = \begin{bmatrix} s_{0} & 0
   & \dots & 0 & 0 \\
   s_1 & s_{0} & \dots & 0 &
   0 \\ \vdots & \vdots& \ddots& \vdots& \vdots \\
   s_N & s_{N-1} & \dots & s_{0} &  0
   \end{bmatrix}.
\end{equation*}

To proceed we define the \emph{characteristic polynomial}
\begin{equation}\label{CharPoly}
  P(w,z)=\det \bigl(\mu(z)-w\cdot I\bigr).
\end{equation}
The equation $P(w,z)=0$ has a form of the standard equation of a
hyperelliptic curve of genus $N+1$
\begin{equation}\label{CharPolyNLS}
  P(w,z)=w^2-(h_0 + h_1 z + \cdots + h_{2N+2}z^{2N+2})=0.
\end{equation}

On the orbit $\mathcal{O}_1^N$ we have $h_{\nu}=c_{\nu}$,
$\nu=N+1,\dots, 2N+1$. Denote by $(w_k,z_k)$ a root of $P(w,z)$ on
the orbit, that is
\begin{equation}\label{HyperCurveNLS}
  w_k^2 = h_0 + h_1 z_k + \cdots h_{N} z^N_k +
  c_{N+1} z^{N+1}_k + \cdots c_{2N+2} z^{2N+2}_k.
\end{equation}
We proceed to show that the set $\{(w_k,z_k)\}$, $k=0,1,\dots,N+1$
defines another parameterization of the orbit $\mathcal{O}_1^N$. We
have to find the explicit relation between the sets
$\{(w_1,z_1),\dots, (w_{N+1},z_{N+1})\}$ and
$\{t_0,t_1,\dots,t_{N},s_0,s_1,\dots, s_N\}$.

Solving \eqref{HyperCurveNLS} for Hamiltonians $h_0, h_1, \dots,
h_{N}$ one gets
\begin{equation}\label{hSystNLS}
  \begin{array}{l} h_{0} = \displaystyle\frac{1}{W}
  [W_1\left(\textstyle w^2\right)
  -c_{N+1} W_1(z^{N+1}) - \cdots - c_{2N+2}
  W_1(z^{2N+2})]\\
  h_{1} = \displaystyle\frac{1}{W}
  [W_2\left(\textstyle w^2\right)
  -c_{N+1} W_2(z^{N+1}) - \cdots - c_{2N+2}
  W_2(z^{2N+2})]\\ \hdotsfor{1} \\
  h_{N} = \displaystyle\frac{1}{W}
  [W_{N+1}\left(\textstyle w^2\right)
  -c_{N+1} W_{N+1}(z^{N+1}) - \cdots - c_{2N+2}
  W_{N+1}(z^{2N+2})],
  \end{array}
\end{equation}
where $W$ and $W_{i}(f(z,w))$ denote the same as in \eqref{hSyst}.

On the orbit $\mathcal{O}_1^N$ the formulas \eqref{HamiltNLS} and
\eqref{hSystNLS} define the same set of functions. We see that both
\eqref{HamiltNLS} and \eqref{hSystNLS} are linear in $c_{\nu}$,
$\nu=N+1,\dots, 2N+2$. As $\{c_{\nu}\}$ is the set of
\emph{independent} parameters one can equate the corresponding
terms. Namely, we obtain
\begin{equation*}
   \frac{s_{m}}{s_{N+1}}=\frac{W_{m+1}(z^{N+1})}{W},\qquad
     m=0,1, \dots,N.
\end{equation*}
This implies that the set $\{z_k\}$ is the set of roots of the
polynomial $s(z)$
\begin{equation*}
  s(z_k)=0,
\end{equation*}
while the variables $\{w_k\}$ satisfy the equalities
\begin{equation*}
  w^2_k=t^2(z_k)+ s(z_k)r(z_k) = t^2(z_k),\qquad k=1,\dots, N+1.
\end{equation*}

\begin{theorem}\label{T:SVNLS}
Suppose the orbit $\mathcal{O}_1^N$ has the coordinates $(t_{m},
s_{m})$, $m=0, 1, \ldots, N$, as above. Then the new coordinates
$(z_k, w_k)$, $k=1,\dots,N+1$, defined by the formulas
\begin{equation}
  s(z_k)=0,\qquad
  w_k = \varepsilon t(z_k),
  \qquad \text{where} \quad \varepsilon^2=1, \label{newvar2}
\end{equation}
have the following properties:
\begin{enumerate}
\item[\textup{(1)}]  a pair $(w_k,z_k)$ is a root of the characteristic
polynomial \eqref{CharPolyNLS}.
\item[\textup{(2)}] a pair
$(z_k, w_k)$ is canonically conjugate with respect to the
Lie-Poisson bracket \eqref{LiePoissonBraNLS1}:
\begin{equation}\label{PoissonBra3}
  \{z_k,z_l\}_1 = 0, \qquad
  \{z_k, w_l\}_1 = \varepsilon \delta_{kl}, \qquad \{w_k,w_l\}_1=0;
\end{equation}
\item[\textup{(3)}] the corresponding  Liouville 1-form is
\begin{align*}
&\Omega_{-1}= \sum\limits_{k} \varepsilon w_{k}\,dz_{k}.&
\end{align*}
\end{enumerate}
\end{theorem}
\begin{proof}
(1) The assertion is a direct consequence of \eqref{CharPolyNLS} and
\eqref{newvar2}.

(2) It is evident that
\begin{equation*}
 \{z_k,z_l\}_1=0,
\end{equation*}
since $z_k$, $k=1,\dots, N+1$ depend only on $s_m$, $m=0,1,\dots,N$
and $s_m$ mutually commute.

Let us calculate the brackets $\{z_k,w_l\}_1$ and $\{w_k,w_l\}_1$.
From \eqref{newvar2} we have
\begin{equation*}
\frac{\partial z_k}{\partial t_{n}}=0,\qquad
  \frac{\partial z_k}{\partial s_{m}}=
  -\frac{z_k^m}{s'(z_k)},\qquad
  \frac{\partial w_l}{\partial t_{n}} = \varepsilon z_l^{n},\qquad
  \frac{\partial w_l}{\partial s_{m}} = \varepsilon t'(z_l)\frac{\partial z_l}{\partial s_m}.
\end{equation*}
Further $\{s_{m},t_{n}\}_1=- s_{m+n}$ when $m+n\leqslant N$ and
$\{s_{m},t_{n}\}_1=0$ when $m+n > N$. Thus, we obtain
\begin{equation*}
  \{z_k, w_l\}_1 = \frac{\varepsilon}{s'(z_k)}
   \frac{s(z_k)-s(z_l)}{z_k-z_l},\qquad \{w_k, w_l\}_1 =
   \left(\frac{t'(z_k)}{s'(z_k)} - \frac{t'(z_l)}{s'(z_l)}\right)
  \frac{s(z_k)-s(z_l)}{z_k-z_l}.
\end{equation*}
Thus,
\begin{equation*}
  \{z_k, w_l\}_1 = \varepsilon \delta_{kl},\qquad \{w_k, w_l\}_1=0.
\end{equation*}

(3) From \eqref{PoissonBra3} it follows that Liouville 1-form on the
orbit $\mathcal{O}_1^N$ is
\begin{equation*}
  \Omega_{-1} = \sum_k \varepsilon w_k \, d z_k.
\end{equation*}

The reduction to Liouville torus is done by fixing the values of
Hamiltonians $h_0$, $h_1$, \ldots, $h_{N}$. On the torus $w_k$ is
the algebraic function of $z_k$ due to (\ref{CharPolyNLS}). After
the reduction the form $\Omega_{-1}$ becomes a sum of meromorphic
differentials on the Riemann surface $P(w,z)\,{=}\,0$.
\end{proof}

Given a set of pairs $(z_k, w_k)$, $k=1,\dots, N+1$, one can find
the set $(t_m, s_m)$, $m=0, \dots, N$, such that the equations
(\ref{newvar2}) are satisfied. Thus, we can define a homomorphism
\begin{equation}\label{homomorfizmNLS}
  \Complex^{2N+2} \to \mathcal{O}_1^N
\end{equation}
that takes the set of pairs $(z_k, w_k)$, $k=1,\dots, N+1$, to a
point of the orbit $\mathcal{O}_1^N$. Since the orbit
$\mathcal{O}_1^N$ is topologically trivial, this implies that the
map (\ref{homomorfizmNLS}) is global.

After the reduction (\ref{homomorfizmNLS}) turns into the map that
takes the  $(N+1)$th symmetric power of Riemann surface to Loiuville
torus:
\begin{equation*}
\Sym\{\mathcal{R}\times \mathcal{R}\times \cdots \times
\mathcal{R}\} \mapsto T^{N+1}.
\end{equation*}

\section{Separation of variables for Heisenberg magnetic chain}
Consider the orbit $\mathcal{O}_2^{N+2}$, $\dim \mathcal{O}_2^{N+2}
= 2(N+2)$. We chose to parameterize the orbit $\mathcal{O}_2^{N+2}$
by the variables $\{\gamma_m,\alpha_m\}$, $m=0,1,\dots, N+1$, that
is we eliminate the set $\{\beta_m\}$, which corresponding basis
elements are nilpotent.

From the orbit equation we find
\begin{equation}\label{betaSystH}
  \beta_{m} = \sum_{j=0}^{N+1} (\widetilde{\Gamma}^{-})^{-1}_{mj}
  (c_{j}-\widetilde{A}_{j}),\qquad  m=0,\,\dots\, N+1,
\end{equation}
where
\begin{equation*}
   \widetilde{\Gamma}^{-}=\begin{bmatrix} \gamma_{0} & 0
   & \dots & 0 & 0 \\
   \gamma_1 & \gamma_0 & \dots & 0 & 0 \\
   \vdots & \vdots& \ddots& \vdots& \vdots \\
   \gamma_{N} & \gamma_{N-1} & \dots & \gamma_{0} & 0 \\
   \gamma_{N+1} & \gamma_{N} & \dots & \gamma_{1} &  \gamma_{0}
   \end{bmatrix} \qquad \text{and} \qquad \widetilde{A}_{\nu} =
   \sum_{\substack{m+n=\nu,\\ 0\leqslant m,n \leqslant N+1}} \alpha_{m}\alpha_{n}.
\end{equation*}
Now, using the parameterization (\ref{betaSystH}), we find
expressions for the Hamiltonians $h_{N+2}$, $h_{N+3}$, \ldots,
$h_{2N+2}$
\begin{equation}\label{HamiltHM}
   h_{n+N+1} = \sum_{m,j=0}^{N+1}\widetilde{\Gamma}^{+}_{nm}
   (\widetilde{\Gamma}^{-})^{-1}_{mj} (c_{j}-\widetilde{A}_{j}) + \widetilde{A}_{n+N+1},
   \quad n=1,\, \dots N+1,
\end{equation}
where
\begin{equation*}
   \widetilde{\Gamma}^{+} = \begin{bmatrix} 0 & \gamma_{N+1} & \dots & \gamma_2 & \gamma_1 \\
   0 & 0 & \dots & \gamma_3 & \gamma_2 \\
   \vdots & \vdots& \ddots& \vdots& \vdots \\
   0 & 0 & \dots & 0 & \gamma_{N+1}
   \end{bmatrix}.
\end{equation*}
Note that the expressions (\ref{HamiltHM}) are linear in $c_{\nu}$,
$\nu=N, \ldots, 2N+1$.

To proceed we use the same characteristic polynomial
\eqref{CharPoly}

On the orbit $\mathcal{O}_2^{N+1}$ we have $h_{\nu}$,
$\nu=0,1,\dots, N+1$. Denote by $(w_k,z_k)$ a root of $P(w,z)$ on
the orbit, that is
\begin{equation}\label{HyperCurveHM}
  w_k^2 = c_0 + c_1 z_k + \cdots c_{N+1} z^{N+1}_k +
  h_{N+2} z^{N+2}_k + \cdots h_{2N+2} z^{2N+2}_k.
\end{equation}
The set $\{(w_k,z_k)\}$, $k=0,1,\dots,N+1$ is insufficient to
parameterize $\mathcal{O}_2^{N+1}$. Let us fix $\alpha_{N+1}$ and
$\gamma_{N+1}$ regarded as Hamiltonians and consider below the
reduced orbit $\mathcal{O}_{2red}^{N+1}$. If we find an explicit
relation between the sets $\{(w_1,z_1),\dots, (w_{N+1},z_{N+1})\}$
and $\{\alpha_0,\alpha_1$, \ldots, $\alpha_{N}$, $\gamma_0$,
$\gamma_1$, \ldots, $\gamma_N\}$ we show that the set
$\{(w_k,z_k)\}$, $k=0,1,\dots,N+1$ defines another parameterization
of the orbit $\mathcal{O}_{2red}^{N+1}$.

\begin{theorem}\label{T:SVHM}
Suppose the orbit $\mathcal{O}_{2red}^{N+1}$ has the coordinates
$(\alpha_{m}, \gamma_{m})$, $m=0, 1, \ldots, N$, as above. Then the
new coordinates $(z_k, w_k)$, $k=1,\dots,N+1$, defined by the
formulas
\begin{equation}
  \gamma(z_k)=0,\qquad
  w_k = \varepsilon \alpha(z_k),
  \qquad \text{where} \quad \varepsilon^2=1, \label{newvar3}
\end{equation}
have the following properties:
\begin{enumerate}
\item[\textup{(1)}]  a pair $(w_k,z_k)$ is a root of the characteristic
polynomial \eqref{CharPolyNLS}.
\item[\textup{(2)}] a pair
$(z_k, w_k)$ is quasi-canonically conjugate with respect to the
Lie-Poisson bracket \eqref{LiePoissonBraNLS2}:
\begin{equation}\label{PoissonBra4}
  \{z_k,z_l\}_2 = 0, \qquad
  \{z_k, w_l\}_2 = -\varepsilon z_k^{N+2} \delta_{kl}, \qquad \{w_k,w_l\}_2=0;
\end{equation}
\item[\textup{(3)}] the corresponding  Liouville 1-form is
\begin{align*}
&\Omega_{N+1}=-\sum\limits_{k} \varepsilon z_k^{-(N+2)}
w_{k}\,dz_{k}.&
\end{align*}
\end{enumerate}
\end{theorem}

\begin{proof}
(1) The assertion is a direct consequence of \eqref{CharPolyNLS} and
\eqref{newvar3}.

(2) It is evident that
\begin{equation*}
 \{z_k,z_l\}_2=0,
\end{equation*}
since $z_k$, $k=1,\dots, N+1$ depend only on $\gamma_m$,
$m=0,1,\dots,N$ and $\gamma_m$ mutually commute.

Let us calculate the brackets $\{z_k,w_l\}_2$ and $\{w_k,w_l\}_2$.
From \eqref{newvar3} we have
\begin{equation*}
\frac{\partial z_k}{\partial \alpha_{n}}=0,\qquad
  \frac{\partial z_k}{\partial \gamma_{m}}=
  -\frac{z_k^m}{\gamma'(z_k)},\qquad
  \frac{\partial w_l}{\partial \alpha_{n}} = \varepsilon z_l^{n},\qquad
  \frac{\partial w_l}{\partial \gamma_{m}} = \varepsilon \alpha'(z_l)\frac{\partial z_l}{\partial \gamma_m}.
\end{equation*}
Further $\{\gamma_{m},\alpha_{n}\}_2=- \gamma_{m+n-N-1}$ when
$m+n\leqslant N+1$ and $\{\gamma_{m},\alpha_{n}\}_2=0$ when $m+n >
N+1$. Thus, we obtain
\begin{gather*}
  \{z_k,w_l\}_2 = \frac{1}{\gamma'(z_k)}
  \frac{z_k^{N+2}\gamma(z_l) -
  z_l^{N+2}\gamma(z_k)}{z_k-z_l},\\
  \{w_k,w_l\}_2 =\left(\frac{1}{\gamma'(z_k)}-\frac{1}{\gamma'(z_l)}\right)
  \frac{z_k^{N+2}\gamma(z_l)-z_l^{N+2}\gamma(z_k)}{z_k-z_l},
\end{gather*}
whence
\begin{equation*}
  \{z_k,w_l\}_2 =- z^{N+2}_k \delta_{kl},\qquad  \{w_k,w_l\}_2 =0.
\end{equation*}

(3) From \eqref{newvar3} it follows that Liouville 1-form on the
orbit $\mathcal{O}_{2red}^{N+1}$ is
\begin{equation*}
  \Omega_{N+1} = -\sum_k \varepsilon z_k^{-(N+2)} w_k \, d z_k.
\end{equation*}

The reduction to Liouville torus is done by fixing the values of
Hamiltonians $h_0$, $h_1$, \ldots, $h_{N}$. On the torus $w_k$ is
the algebraic function of $z_k$ due to (\ref{CharPolyNLS}). After
the reduction the form $\Omega_{-1}$ becomes a sum of meromorphic
differentials on the Riemann surface $P(w,z)\,{=}\,0$.
\end{proof}

\begin{remark}
When $\gamma_{N+1}=0$, $\alpha_{N+1}\neq 0$ one can replace
$\alpha_m, \beta_m, \gamma_m$ by $t_m, s_m, r_m$, $m=0,1,\dots,N$,
according to (\ref{PreviatoCoord}). In this case the new coordinates
$(z_k, w_k)$, $k=1,\dots,N+1$, defined by the formulas
\begin{equation}
  s(z_k)=0,\qquad
  w_k = \varepsilon t(z_k),
  \qquad \text{where} \quad \varepsilon^2=1. \label{newvar4}
\end{equation}
and have the same properties as in Theorem 4.
\end{remark}

The results of Theorem~\ref{T:SVNLS} and \ref{T:SVHM} can be
summarized as follows. Liouville tori for the nonlinear
Schr\"{o}dinger equation and Heisenberg magnetic chain have the same
number of parameterizing variables $z_k$ and each variable belongs
to the hyperelliptic curve (\ref{CharPolyNLS}) of genus $g=N+1$. In
other words, the common Liouville torus is the Jacobi variety of the
curve (\ref{CharPolyNLS}).

\label{lastpage}

\end{document}